\documentclass[10pt, leqno]{amsart}
\usepackage{amssymb}
\usepackage{amsmath,amscd}
\usepackage[initials,nobysame,alphabetic]{amsrefs}
\usepackage{amsmath, amssymb}
\usepackage{amsfonts}
\usepackage{mathrsfs}
\usepackage{mathpazo}
\usepackage{mathtools}
\usepackage[arrow,matrix,curve,cmtip,ps]{xy}

\usepackage{amsthm}

\usepackage{float}
\usepackage{hyperref}
\usepackage{graphics}
\usepackage{graphicx}
\usepackage{verbatim}
\usepackage{multirow}
\usepackage{tikz}
\usepackage{enumerate}
\allowdisplaybreaks

\newtheorem{theorem}{Theorem}[section]

\newtheorem{proposition}[theorem]{Proposition}
\newtheorem{corollary}[theorem]{Corollary}

\newtheorem*{theorem*}{Theorem}

\newtheorem{definition}[theorem]{Definition}
\newtheorem{example}[theorem]{Example}

\numberwithin{equation}{section}

\newcommand{\R}{\mathbb{R}}
\newcommand{\Ff}{\mathbb{F}}

\newcommand{\F}{\mathcal{F}}

\newcommand{\cS}{\mathcal{S}}

\newcommand{\Prob}{\mathbb{\Prob}}

\newcommand{\mytilde}{\raise.17ex\hbox{$\scriptstyle\mathtt{\sim}$}}

\def\a{\alpha}             
\def\b{\beta}                                    

\def\l{\lambda}

\begin{document}
\title[Nonlinear reserving]{Nonlinear reserving and multiple contract modifications in life insurance
}

\author{Marcus C.~Christiansen  and Boualem Djehiche}

\address{Institute of Mathematics \\ Carl von Ossietzky University, 26111 Oldenburg, Germany} \email{marcus.christiansen@uni-oldenburg.de}
\address{Department of Mathematics \\ KTH Royal Institute of Technology \\ 100 44, Stockholm \\ Sweden} \email{boualem@kth.se}

\date{This version \today}

\subjclass[2010]{60H10, 91B30}

\keywords{Implicit options, prospective reserve, nonlinear semi-Markov chain, Backward SDEs, Thiele equation}

\begin{abstract} 
Life insurance cash flows become reserve dependent when contract conditions are modified during the contract term on condition that actuarial equivalence is maintained. As a result, insurance cash flows and prospective reserves depend on each other in a circular way, and it is a non-trivial problem to solve that circularity and make cash flows and prospective reserves well-defined. In Markovian models, the (stochastic) Thiele equation and the Cantelli Theorem are the standard tools for solving the circularity issue and for maintaining actuarial equivalence. This paper expands the stochastic Thiele equation and the Cantelli Theorem to non-Markovian frameworks and presents a recursive scheme for the calculation of multiple contract modifications. 
\end{abstract}

\maketitle


\section{Introduction}

Life insurance products typically comprise implicit options. This involves guaranteed components as well as rights to modify contract conditions during the contract term, see e.g.~Gatzert (2009) for an overview. In recent years  insurers and regulators paid increasing attention to the proper pricing and reserving for contracts with implicit options. 
In the actuarial literature there are numerous papers on market evaluations of implicit financial guarantees, but the mathematical modelling of premium payment modifications and modifications of insurance coverage is still underdeveloped. This paper helps to close that gap. 

The prospective reserve of a life insurance contract is defined as the conditional expectation of the aggregated and discounted future insurance cash flow given the currently available information. Traditionally, the insurance cash flow is defined first, and then the prospective reserve is defined and calculated on the basis of that cash flow. However, in case that the insurance cash flow depends also on the prospective reserve, then we have a circular structure and the classical definition of the prospective reserve becomes an implicit equation for which existence and uniqueness of a solution are  in general unclear. In a Markovian framework Djehiche \& L\"{o}fdahl (2016) showed that the circularity problem is equivalent to solving a  backward stochastic differential equation (BSDE). In case that the cash flow satisfies certain Lipschitz conditions, then the BSDE has a unique solution and the prospective reserve and the cash flow are well-defined. In this paper we generalize that concept to a non-Markovian framework, which automatically includes popular semi-Markovian models. 

In Djehiche \& L\"{o}fdahl (2016) the cash flow at a certain time point may depend on the reserve at the same time but not on the reserve at earlier time points. This restriction is fine when we model surrender options, but it excludes various other modifications of premium payments and insurance coverage. For example, think of a  free policy option where the insurance cash flow after exercising the option depends on the reserve at the time of exercising the option.  One possibility is to adhere to the BSDE approach,  but the mild Lipschitz conditions that Djehiche \& L\"{o}fdahl (2016) use have to be replaced with much more restrictive Lipschitz conditions that are usually not satisfied in practice. Instead, in this paper we suggest a recursive scheme that runs forward in time through the contract modifications. Our results differ from the existing literature in two ways: First, we allow for an unbounded number of contract modifications. Second, we avoid any kind of Markov assumptions.  

The most widely studied kinds of reserve dependent insurance cash flows are surrender payments upon lapse, see e.g.~M{\o}ller \& Steffensen (2007) and references therein. In case that the Cantelli Theorem applies, surrender may be simply ignored, see e.g.~Milbrodt \& Stracke (1997). If actuarial equivalence is not fully maintained but the dependence on the reserve is linear, explicit formulas are still within reach, see Christiansen et al.~(2014).  If the dependence is not necessarily linear but at least Lipschitz continuous, then the  BSDE concept of Djehiche \& L\"{o}fdahl (2016) gives a general answer on how to define and calculate reserves in the presence of lapse. Here we further expand the results of Djehiche \& L\"{o}fdahl (2016) to non-Markovian models. 

The second most studied option is the free policy option,  see e.g.~Hendriksen et al.~(2014), Buchardt et al.~(2015) and Buchardt \& M{\o}ller (2015). Here, 
based on the Cantelli Theorem, an adjustment factor is applied on the life insurance cash flow at the contract modification such that  actuarial equivalence is maintained. We generalize that concept to non-Markovian models. Furthermore, we allow for arbitrarily many contract modifications of all kinds.  

The paper is organized as follows: In section 2 we define the state dynamics of a life insurance policy and its corresponding life insurance cash flow. We also show the link to martingale theory, which becomes relevant in the sections to follow. Section 3 introduces the prospective reserve as the solution of a  backward stochastic differential equation and extends the results of Djehiche \& L\"{o}fdahl (2016) to non-Markovian frameworks. In section 4 we add the possibility of an unbounded number of contract modifications and discuss the definition and calculation of prospective reserves under actuarial equivalence conditions.  

\section{Life insurance policy modeling}\label{Section-LifeInsurPolicyModelling}
\subsection{State dynamics of a life insurance policy} \label{Subsection_2_1}

In the multi-state framework within life insurance on some   finite state space $\cS\subset \mathbb{N}_0$, the evolution of an insurance policy  is usually described by an $\cS$-valued c{\`a}dl{\`a}g (right continuous with left limits) pure jump process $X$, starting at a deterministic state $X(0)=x_0\in\cS$, defined on the completed filtered probability space $(\Omega,\F^0,\Ff^0=(\F^0_t)_{t\ge 0},P)$, where $\Ff^0$ is the completed natural filtration of $X$. The filtration  $\Ff^0$ satisfies the usual conditions since the natural filtration of $X$ is always right-continuous, see Theorem 2.2.4 in Last \& Brandt (1995). 

To $X$ we associate the indicator process $I_i(t)=\mathbf{1}_{\{X(t)=i\}}$ whose value is $1$ if $X$ is in state $i$ at time $t$ and $0$ otherwise, and the counting processes defined by
$$
N_{ij}(t):=\#\{s\in(0,t]:X(s-)=i, X(s)=j\},\quad N_{ij}(0)=0,
$$
which count the number of jumps from state $i$ into state $j$ during the time interval $(0,t]$.
Since $X$ is c{\`a}dl{\`a}g, $I_i$ and $N_{ij}$ are c{\`a}dl{\`a}g as well. Moreover, by the relationship
\begin{equation*}
X(t)=\sum_i i\,I_i(t),\quad I_i(t)=I_i(0)+\underset{j: j\neq i}\sum\left(N_{ji}(t)-N_{ij}(t)\right),
\end{equation*}
the state process, the indicator processes, and the counting processes carry the same information which is represented by the natural filtration  $\Ff^0$ of $X$.
Let $0=T_0<T_1<T_2<\ldots $ denote the jump times of the process and
$$
N(t):=\sum_{n=1}^{\infty}\mathbb{1}_{\{T_n\le t\}}=\sup\{n, \, T_n\le t\}=\sum_{i,j : i\ne j} N_{ij}(t).
$$
For each $t\ge 0$, let $U(t)$ be the time spent in the current state $X(t)$, i.e.
\begin{equation*}
U(t)=t-T_{N(t)},\qquad i.e. \,\,\, U(t)=t-T_n,\,\,\mbox{if}\,\,  T_{n}\le t<T_{n+1},\,\, n \in \mathbb{N}.
\end{equation*}
%
Two popular models for the pure jump process have been considered in the literature, cf.~Christiansen (2012):
\begin{example}[Markov models]
The process  $X$ is assumed to be Markovian.
\end{example}

\begin{example}[semi-Markov models]
The process $\tilde{X}:=(X,U)$ is assumed to be Markovian.
\end{example}

\subsection{Associated martingales} 
In the sequel, we make the following standing assumption:
\begin{itemize}
\item[(A1)] The compensators of the counting processes  $N_{ij}(t)$, $i,j \in \mathcal{S}$, $i \neq j$, have Lebesgue-densities $I_i(t-)\l_{ij}(t)$, $i,j \in \mathcal{S}$, $i \neq j$, that  satisfy
\begin{equation*}
E\left[\int_0^T\sum_{i,j: i\neq j} I_i(t-)\l_{ij}(t)dt\right]<\infty.
\end{equation*}
\end{itemize}
We denote the processes $(\l_{ij})_{ij}$ as jump intensities (or transition intensities). If $X$ is a Markov process, then the jump intensities $(\l_{ij})_{ij}$ are deterministic, whereas in general they are predictable processes.
Mimicking the proof of  Lemma 21.13 in Rogers \& Williams (2000), the compensated processes associated with the counting processes $N_{ij}$, defined by
 \begin{align}\label{mart-0}
M_{ij}(t)=N_{ij}(t)-\int_0^t I_i(s-)\l_{ij}(s)\, ds,\quad M_{ij}(0)=0,
\end{align}
are zero mean, square integrable and mutually orthogonal $P$-martingales. 
We call  $M:=\{M_{ij},\, i\neq j\}$ the accompanying martingale of the counting process $N:=\{N_{ij},\, i\neq j\}$ or of the process $X$.
Let $\{Z_{ij}, \, i\neq j\}$ be a family of predictable processes and set
\begin{equation*}
\|Z(t)\|^2_{\Lambda}:=\sum_{i,j:i\neq j}Z^2_{ij}(t)I_i(t-)\l_{ij}(t),\quad 0< t\le T.
\end{equation*}
The local martingale
\begin{equation*}
\int_{(0,t]}Z(s)dM(s):=\sum_{i,j:  i\neq j}\int_{(0,t]} Z_{ij}(s)dM_{ij}(s)
\end{equation*}
is a square-integrable martingale if
\begin{equation*}
E\left[\int_0^T \|Z(s)\|^2_{\Lambda}ds\right]<\infty
\end{equation*}
since the following Doob inequality holds:
\begin{equation*}
E\left[\sup_{0\le t\le T}\left|\int_{(0,t]}Z(s)dM(s)\right|^2\right]\le 4 E\left[\int_0^T \|Z(s)\|^2_{\Lambda}ds\right].
\end{equation*}
Since $X(0)$ is deterministic and  the filtration $\mathbb{F}^0$ generated by  $X$ is the
same as the filtration generated by the family of counting processes $\{N_{ij},\, i\neq j\}$, we state the following martingale representation theorem (see e.g.~Br\`{e}maud (1981), Theorem T11 or Rogers \& Williams (2000), IV-21, Theorem 21.15).

\begin{proposition}[Martingale representation theorem]\label{mart-rep}  If $Y$ is a (right-continuous) square-integrable $\mathbb{F}^0$-martingale, there exists a  unique ($dP\times I_i(s-) \l_{ij}(s)ds$-almost everywhere) family of predictable processes $Z_{ij}, \, i\neq j,$ satisfying
\begin{equation}\label{Z-g-int}
E\left[\int_0^T \|Z(s)\|^2_{\Lambda}ds\right]<\infty
\end{equation}
such that
\begin{equation}\label{mart-rep-L}
Y(t)=Y(0)+\int_{(0,t]}Z(s)dM(s), \quad 0\le t\le T.
\end{equation}
\end{proposition}
In fact the form of the process $Z$ can be made explicit as shown in the following Proposition.
\begin{proposition}[Explicit martingale representation]\label{mart-rep-explicit}
Let $\zeta$ be an integrable random variable. The unique right-continuous process  $Y$ defined by  $Y(t):= E[\zeta| \mathcal{F}^0_t]$, $t \geq 0$, satisfies
\eqref{mart-rep-L} for $Z$ defined as
\begin{align*}
  Z_{ij}(t) := \sum_{n=0}^{\infty}  \mathbb{1}_{\{T_n < t \leq T_{n+1}\}} \left( E [ \zeta | \mathcal{F}^0_{T_n}, T_{n+1}=t, X(T_{n+1})=j] - \frac{ E [ \zeta \mathbb{1}_{\{ T_{n+1}>t\}} | \mathcal{F}_{T_n}^0]}{E [  \mathbb{1}_{\{  T_{n+1}>t\}} | \mathcal{F}_{T_n}^0]} \right).
\end{align*}
\end{proposition}

\begin{proof}
First of all, suppose that $X$ has at most one jump. Then, according to Chou \& Meyer (1975) we have
\begin{align*}
 Y(t) = Y(0) +\sum_{j: j\neq x_0} \int_{(0,t]}  \left( E [ \zeta | \mathcal{F}_{0}^0, T_{1}=s, X(T_{1})=j] - \frac{ E [ \zeta \mathbb{1}_{\{ s < T_{1}\}} | \mathcal{F}_{0}^0]}{E [  \mathbb{1}_{\{ s < T_{1}\}} | \mathcal{F}_{0}^0]} \right)d M_{x_0j}(s)
\end{align*}
almost surely for each $t >0$.
The statement remains true for any enlargement of the initial information $\mathcal{F}_{0}^0$, see Chou \& Meyer (1975).  Following the construction in Elliott (1976),  by applying the single jump result on the inter-arrival times $S_{n+1}:=T_{n+1}-T_{n}$, we can show that
\begin{align*}
  Y(T_n+t)- Y(T_{n})
  &=   \sum_{i,j:i\neq j}\int_{(T_{n},T_n+t]}  \Big( E [ \zeta | \mathcal{F}_{T_{n}}^0, S_{n+1}=  s-T_n, X(T_n+ S_{n+1})=j] \\
  & \qquad  \qquad  \qquad \qquad \qquad \qquad  -  \frac{ E [ \zeta \mathbb{1}_{\{ T_{n+1}>s\}} | \mathcal{F}_{T_n}^0]}{E [  \mathbb{1}_{\{  T_{n+1}>s\}} | \mathcal{F}_{T_n}^0]}   \Big)d M_{ij}(s)
\end{align*}
for $T_{n}<t\leq T_{n}+S_{n+1}$, $n \in \mathbb{N}_0$.
\end{proof}

\begin{corollary}\label{Corollary_explicit-mart-repres}
Let $(\zeta(t))_{t \geq 0}$ be an integrable c\`{a}dl\`{a}g process such that $\zeta(t)-\zeta(0)$ is $\mathcal{F}^0_t$-measurable for each $t \geq 0$. Then the unique right-continuous process  $Y$ defined by  $Y(t):= E[\zeta(t)| \mathcal{F}^0_t]$, $t \geq 0$, satisfies
\begin{equation}\label{mart-rep-dynamic}
Y(t)= \zeta(t)-\zeta(0) +\int_{(0,t]}Z(s)dM(s),\quad  t\ge 0,
\end{equation}
for $Z$ defined as
\begin{align*}
  Z_{ij}(t):= \sum_{n=0}^{\infty}  \mathbb{1}_{\{T_n < t \leq T_{n+1}\}} &\Big( E [ \zeta(t-) | \mathcal{F}^0_{T_n}, T_{n+1}=t, X(T_{n+1})=j] \\
  & \qquad\qquad - \frac{ E [ \zeta(t-) \mathbb{1}_{\{ T_{n+1}>t\}} | \mathcal{F}_{T_n}^0]}{E [  \mathbb{1}_{\{  T_{n+1}>t\}} | \mathcal{F}_{T_n}^0]}  \Big).
\end{align*}
\end{corollary}

\begin{proof}
By applying Proposition \ref{mart-rep-explicit} on the martingale
\begin{align*}
  E[ \zeta(t) | \mathcal{F}^0_t ]  - (\zeta(t) - \zeta(0)) = E[ \zeta(0) | \mathcal{F}^0_t ]
\end{align*}
and using that
\begin{align*}
 &\sum_{n=0}^{\infty}  \mathbb{1}_{\{T_n < t \leq T_{n+1}\}}\left( E [\zeta(0) | \mathcal{F}^0_{T_n}, T_{n+1}=t, X(T_{n+1})=j] -\frac{ E [ \zeta(0) \mathbb{1}_{\{ T_{n+1}>t\}} | \mathcal{F}_{T_n}^0]}{E [  \mathbb{1}_{\{  T_{n+1}>t\}} | \mathcal{F}_{T_n}^0]} \right) \\
 &=   \sum_{n=0}^{\infty}  \mathbb{1}_{\{T_n < t \leq T_{n+1}\}}\left(E [ \zeta(t-) | \mathcal{F}^0_{T_n}, T_{n+1}=t, X(T_{n+1})=j] - \frac{ E [ \zeta(t-) \mathbb{1}_{\{ T_{n+1}>t\}} | \mathcal{F}_{T_n}^0]}{E [  \mathbb{1}_{\{  T_{n+1}>t\}} | \mathcal{F}_{T_n}^0]} \right)
\end{align*}
almost surely for each $t>0$ since $\zeta(t-)-\zeta(0)$ is $\mathcal{F}^0_{t-}$-measurable, we end up with equation \eqref{mart-rep-dynamic}.
\end{proof}

\subsection{Life insurance cash flow}\label{Section_Cashflow}
A standard life insurance payment process $A(t)$ of accumulated contractual benefits less premiums payable during the time interval $[0, t]$ is of the form
$$
dA(t)=\sum_iI_i(t-) \Big( \alpha_i(t)\,dt+ a_i(t)\, d \nu(t) \Big)+\sum_{i,j : i\ne j} \b_{ij}(t)\,dN_{ij}(t),\quad A(0-)=0,
$$
where $\alpha_i$  is a predictable process specifying continuous sojourn payments in state $i$, $a_i$ is a predictable process specifying lump sum sojourn payments in state $i$, and $\nu$ is a deterministic  step function with step heights of $+1$ and at most a finite number of steps on compact intervals. 
 Furthermore, $\b_{ij}$ is a predictable process that specifies a transition payment due immediately upon a transition from state $i$ to state $j$ at time $t$.  In the following we will also use the short notation
 \begin{align*}
\alpha_{X(t-)}(t) &=  \sum_iI_i(t-)\alpha_i(t),\\
a_{X(t-)}(t) &=  \sum_iI_i(t-)a_i(t).
\end{align*}
We generally assume that there is a finite maximum contract time $T < \infty$, i.e.~ $\alpha_i(t)=0$, $a_i(t)=0$ and $\b_{ij}(t)=0$ for all $t >T$ and  $i,j \in \mathcal{S}$, $i \neq j$.

\subsection{Life insurance model for given jump intensities}

Instead of starting from a probability space with probability measure $P$ and then identifying the transition intensities $(\l_{ij})_{ij}$, practitioners usually prefer the reversed approach that starts from a given set of transition intensities and then identifies the corresponding probability measure. While it is a well known fact in the actuarial literature that the transition intensities and the starting value of  a Markov process $X$ uniquely determine the whole distribution of $X$, we want to point out here that the reversed approach still works for non-Markovian models. 
\begin{proposition}[Existence of jump process distributions for given transition intensities]\label{ModelFromCompensators}
Given a sample space $\Omega$ and a jump process $X: [0,\infty) \times \Omega \rightarrow \mathcal{S}$, suppose that  $f_{ij}: [0,\infty) \times \Omega  \rightarrow [0,\infty)$, $i,j \in \mathcal{S}, i\neq j$,  are predictable processes with respect to the (non-completed) natural filtration of $X$, and for each $\omega \in \Omega$ let 
\begin{align*}
    \int_0^T\sum_{i,j: i\neq j}f_{ij}(t,\omega)dt  < \infty.
\end{align*}
Then there exists a unique probability measure $P$ such that the transition intensities $(\l_{ij})_{ij}$ of $X$ satisfy 
$$ I_i(t-)\l_{ij}(t) = I_i(t-) f_{ij}(t,\cdot), \quad t \geq 0,\, i,j \in \mathcal{S},\, i\neq j.$$
\end{proposition}
\begin{proof}
By applying Theorem 8.2.1 and Theorem 8.2.3 from Last \& Brandt (1995), we get a unique probability measure $P$ on $\sigma(X(s):s\geq 0)$  such that  the compensators of  the counting processes $N_{ij}$ have Lebesgue-densities $I_i(t-)f_{ij}(t,\cdot) $ w.r.t.~the non-completed natural filtration of $X$. By completing the measure $P$, the sigma-algebra $\sigma(X(s):s\geq 0)$ and the natural filtration of $X$, we obtain that the processes $I_i(t-) f_{ij}(t,\cdot) dt $ are still the compensators of $d N_{ij}(t) $ under the $P$-completed natural filtration of $X$.
\end{proof}

\section{Prospective reserves}\label{SectionProspectiveReserves}

\medskip Following Norberg (1991, 1992), we recall the conditional expectation formulation of the prospective reserve for the above life insurance policy, given the compensators $\Lambda=(\l_{ij})_{ij}$ and a discount rate $\delta$. We  assume that $\delta$ is a bounded and  progressively measurable process.

\begin{definition}
The prospective reserve associated with the payment process $A$, the matrix $\Lambda$ and discount rate $\delta$ is
\begin{equation}\label{reserve-1}
Y(t):=E\Big[\int_{(t,T]} e^{-\int_t^s \delta(u)du}\,dA(s)\Big|\F^0_t\Big],
\end{equation}
where the pair $(\Lambda,\delta)$ is called the basis of the prospective reserve.
\end{definition}
The first two main results of this paper displayed below are Propositions \ref{bsde-1} and \ref{Nonlinear-NonMarkovian-StochThiele}. In Proposition \ref{bsde-1} we give a BSDE formulation of the prospective reserve $Y$ when the involved payment processes $\alpha_i,a_i$ and $\b_{ij}$ are mappings on $[0,\infty) \times \Omega$. 
We call it {\it the linear prospective reserve}. In Proposition \ref{Nonlinear-NonMarkovian-StochThiele} we study the case where the payment processes $\alpha_i(t)$, $a_i(t)$ and $\beta_{ij}(t)$ may additionally depend on the current reserve $Y(t-)$ and the process $Z(t)$, in which case the definition of $Y$ given by \eqref{reserve-1} is circular, and it is not clear if $Y$ is still well-defined, i.e.~it is unclear whether $Y$ exists and whether it is unique. Under mild conditions on the payment processes and the jump intensities of the state process $X$, we show  in Proposition \ref{Nonlinear-NonMarkovian-StochThiele} that the prospective reserve $Y$, coined {\it the nonlinear prospective reserve}, exists and is unique.

\subsection{Linear reserving}\label{Subsection-LinearReserving-Classical}
By linear reserving we mean the case where the payment processes $\alpha_i$, $a_i$ and $\beta_{ij}$ are mappings on $[0,\infty)  \times \Omega$ which do not dependent on the prospective reserve. 
We assume that
\begin{enumerate}
\item [(A2)] the payment processes $(\alpha_i)_i$, $(a_i)_i$ and $\beta=(\beta_{ij})_{ij}$  satisfy
\begin{equation*}
E\left[\int_0^T \left(|\alpha_{X(t-)}(t)|^{2}+\|\b(t)\|^2_{\Lambda}\right)dt  + \int_{(0,T]} |a_{X(t-)}(t)|^{2} d\nu(t) \right]<\infty.
\end{equation*}
\end{enumerate}
Noting that, in view of (A2), the process defined by
\begin{equation*}
d\widetilde M_i(t):=\sum_{j:j\neq i} \b_{ij}(t)\,dM_{ij}(t)
\end{equation*}
is  a square integrable $\mathbb{F}^0$-martingale, the payment process $A$ can be written as
\begin{equation}\label{CompensatorOfA}
dA(t)=\gamma_{X(t-)}(t)\, d t + a_{X(t-)}(t) \,d \nu(t) +\sum_i d\widetilde{M}_{i}(t)
\end{equation}
with
$$
\gamma_i(t):=\alpha_{i}(t)+ \sum_{j:j\neq i} \beta_{ij}(t) I_i(t-)\lambda_{ij}(t),\quad i\in\cS.
$$
%
Using the martingale property of the $\widetilde{M}_i$'s, the prospective reserve \eqref{reserve-1} of the life insurance contract becomes
\begin{equation}\label{eqn:Norberg-Thiele1}
Y(t)=E\Big[\int_{(t,T]} e^{-\int_t^s \delta(u)du}\Big(\gamma_{X(s-)}(s)\,ds +a_{X(s-)}(s)\, d\nu(s)\Big) \Big|\mathcal{F}^0_t\Big],\quad 0\le t\le T,
\end{equation}
which may be written as
$$
e^{-\int_0^t \delta(u)du}Y(t)+\int_{(0,t]} e^{-\int_0^s \delta(u)du}\Big(\gamma_{X(s-)}(s)\,ds + a_{X(s-)}(s)\, d\nu(s)\Big)=\widehat{M}(t),
$$
where $\widehat{M}$ is the square integrable martingale defined by
$$
\widehat{M}(t)=E\Big[\int_{(0,T]} e^{-\int_0^s \delta(u)du}\Big(\gamma_{X(s-)}(s)\,ds +a_{X(s-)}(s)\, d\nu(s)\Big) \Big|\mathcal{F}^0_t\Big].
$$
By the Martingale Representation Theorem, there exists a  unique ($dP\times I_i(s-) \l_{ij}(s)ds$-a.e.) family of predictable processes $Z_{ij}, \, i\neq j,$ satisfying
\begin{equation*}
E\left[\int_0^T \|Z(s)\|^2_{\Lambda}ds\right]<\infty,
\end{equation*}
such that
$$
d\widehat{M}(t)=Z(t)dM(t).
$$
This fact leads to the following BSDE representation of the prospective reserve.
\begin{proposition}[ Backward SDE formulation of the prospective reserve]\label{bsde-1}
The prospective reserve $Y$, given in \eqref{reserve-1}, associated with the payment process $A$, the matrix $\Lambda$ and discount rate $\delta$, satisfies the BSDE
\begin{equation}\label{bsde-1-Y}
dY(t)=\big(-\delta(t)Y(t) + \gamma_{X(t-)} (t)\big)  dt +a_{X(t-)}(t) \, d\nu(t) +Z(t)dM(t), \quad Y(T)=0,
\end{equation}
for a predictable process $Z=(Z_{ij})_{ij}$ such that
\begin{align}\label{expl-form-Z}
  I_i(t-) Z_{ij}(t) = I_i(t-) \bigg(\beta_{ij}(t) +E[ Y(t) |\mathcal{F}^0_{t-}, X(t)=j]- E[ Y(t) |\mathcal{F}^0_{t-}, X(t)=i] \bigg)
\end{align}
almost surely for each $t \geq 0$, $i,j \in \mathcal{S}$, $i\neq j$.
\end{proposition}
\begin{proof}
 By applying Corollary \ref{Corollary_explicit-mart-repres} on the process
 \begin{align*}
   \zeta(t):= \int_{(t,T]} e^{-\int_0^s \delta(u)du}\Big(\gamma_{X(s-)}(s)\,ds +a_{X(s-)}(s)\, d\nu(s)\Big) ,
 \end{align*}
 we obtain that
\begin{align}\label{expl-form-Z2}
  Y(t) e^{-\int_0^t \delta(u)du} = - \int_{(0,t]} e^{-\int_0^s \delta(u)du}\Big(\gamma_{X(s-)}(s)\,ds +a_{X(s-)}(s)\, d\nu(s)\Big) + \int_{(0,t]} \widetilde{Z}(s) d M(s)
\end{align}
for 
$$
Z_{ij}(t) =   E [ \zeta(t-) | \mathcal{F}^0_{T_n}, T_{n+1}=t, X(T_{n+1})=j] - \frac{E [\zeta(t-)  \mathbb{1}_{\{ T_{n+1}>t\}} | \mathcal{F}^0_{T_n}]}{E [ \mathbb{1}_{\{ T_{n+1}>t\}} | \mathcal{F}^0_{T_n}]}
$$ 
on $\{T_n < t \leq T_{n+1}, X_{t-}=i\}$.
Since $d M_{ij}(t) = I_i(t-) d M_{ij}(t)$ for all $i\neq j$ and all $t>0$, for each $t>0$ the processes  $Z_{ij}(t)$ almost surely equal
\begin{align*}
  I_i(t-)Z_{ij}(t) = I_i(t-)\Big(E[ \zeta(t-) |\mathcal{F}^0_{t-}, X(t)=j]- E[ \zeta(t-) |\mathcal{F}^0_{t-}, X(t)=i] \Big).
\end{align*}
Since $\zeta(t)- \zeta(t-)= a(t) \Delta \nu (t)$ is $\mathcal{F}^0_{t-}$-measurable, in the latter formula we may replace $\zeta(t-)$ by $\zeta(t)$.
Integration by parts yields that
\begin{align*}
  d\left(Y(t) e^{-\int_0^t \delta(u)du}\right) =e^{-\int_0^t \delta(u)du} d Y(t) - Y(t) e^{-\int_0^t \delta(u)du} \delta(t) dt,
\end{align*}
 so equation \eqref{expl-form-Z2} can be rewritten to
\begin{align*}
   e^{-\int_0^t \delta(u)du} d Y(t)
=   e^{-\int_0^t  \delta(u)du} \bigg(\Big( \delta(t)Y(t) dt -\gamma_{X(t-)}(t) \Big)  dt -a_{X(t-)}(t) \, d\nu(t)+ Z(t) d M(t)\bigg),
\end{align*}
 using the fact that $E[ \zeta(t) |\mathcal{F}^0_{t-}, X(t)=j] = e^{-\int_0^t \delta(u)du} E[ Y(t) |\mathcal{F}^0_{t-}, X(t)=j]$. Equation \eqref{bsde-1-Y} follows now from the Radon-Nikodym Theorem.
\end{proof}
Note that the BSDE \eqref{bsde-1-Y} differs from the stochastic Thiele equations according to Norberg (1992) and M{\o}ller (1993), since we additionally use the decomposition \eqref{CompensatorOfA}. The latter decomposition has the advantage that \eqref{bsde-1-Y} has a form that is more common in the literature on BSDEs.
\begin{example}[Markov models]
If $X$ is a Markov process, i.e.~a process for which the jump intensities $(\l_{ij})_{ij}$ are deterministic, then 
$$ E[ Y(t) |\mathcal{F}^0_{t-},X(t)]=  E[ Y(t) |X(t)]$$ 
almost surely and the process $Z=(Z_{ij})_{ij}$  can be represented as 
\begin{align*}
   Z_{ij}(t)=   E[  Y(t) |X(t)=j]- E[ Y(t) | X(t)=i].
\end{align*}
Furthermore, if the processes $\alpha_i$, $a_i$ and $\b_{ij}$ are deterministic, it can be shown  that $Y(t)=V(t,X(t))$ and $Z_{ij}(t)= V(t,j)- V(t,i)$ for a deterministic function $V(t,x)$ that solves the Thiele equation, cf.~M{\o}ller (1993), Djehiche \& L\"{o}fdahl (2016).
\end{example}
\begin{example}[semi-Markov models]
If $X$ is a semi-Markov process, i.e.~a process for which the jump intensities have the form $\l_{ij}(t)  =\mu_{ij}(t, U(t))$ for deterministic functions $\mu_{ij}(t,u)$, then 
$$ E[ Y(t) |\mathcal{F}^0_{t-},X(t)]=  E[ Y(t) |U(t), X(t)]$$ 
almost surely and the process $Z=(Z_{ij})_{ij}$ can be represented as 
\begin{align*}
   Z_{ij}(t)=   E[  Y(t) |U(t),X(t)=j]- E[ Y(t) | U(t),X(t)=i].
\end{align*}
Furthermore, if the payment processes are of the form $\alpha_i(t)=\alpha_i(t,U(t))$, $a_i(t)=a_i(t,U(t))$ and  $\b_{ij}(t)=\b_{ij}(t,U(t))$ for deterministic functions $\alpha_i(t,u)$, $a_i(t,u)$ and $\b_{ij}(t,u)$, it can be shown that $Y(t)=V(t,X(t),U(t))$  and $Z_{ij}(t)=V(t,j,U(t))- V(t,i,U(t))$ for some deterministic function $V(t,x,u)$ which solves the semi-Markov Thiele equation, cf.~M{\o}ller (1993).
\end{example}

\subsection{Nonlinear reserving}
By nonlinear reserving we mean the case where the payment processes $\alpha_i(t)$, $a_i(t)$ and $\beta_{ij}(t)$ may  depend on the prospective reserve $Y(t-)$ and the process $Z(t)$,
\begin{align}\label{Nonlinear-Reserving-Payment-Functions}
\begin{split}
\alpha_i(t)(\omega)&:=\alpha_i(t,\omega, Y(t-),Z(t)),\\
a_i(t)(\omega)&:=a_i(t,\omega, Y(t-),Z(t)),\\
\beta_{ij}(t)(\omega)&:=\beta_{ij}(t,\omega,Y(t-),Z(t)).
\end{split}
\end{align}
As a consequence, the definition of $Y$ according to \eqref{reserve-1} is circular, and it is not clear if $Y$ is still well-defined, i.e.~it is unclear whether $Y$ exists and whether it is unique. We will now present rather mild conditions that guarantee that the nonlinear prospective reserve is indeed well-defined.

Assume that the process $\gamma_i (t)(\omega)= \gamma_i(t,\omega,Y(t-),Z(t))$ 
satisfies:
\begin{itemize}
\item[(A3)] There is some real $C\in [ 0,\infty)$ such that $P$-a.s., for all $t\in[0,T]$, $y,\overline{y}\in\mathbb{R}, \,z=(z_{ij}),\,\overline{z}=(\overline{z}_{ij}), z_{ij}, \overline{z}_{ij}\in\mathbb{R}$,
$$
|\gamma_i(t,\omega, y,z)-\gamma_i(t,\omega, \overline{y},\overline{z})|\le C \left(|y-\overline{y}|+\|z-\overline{z}\|_{\Lambda}\right), \quad i \in \mathcal{S}.
$$
\item[(A4)] $E\left[\int_0^T|\gamma_i(t,\omega,0,0)|^2dt\right]<\infty$, $i \in \mathcal{S}$.
\end{itemize}
Furthermore, we make the following assumption:
\begin{itemize}
\item[(A5)] There are reals $C_1\in[0,1)$ and $C_2\in[0,\infty)$ such that $d P \times d \nu$-a.e., for all  $y,\overline{y}\in\mathbb{R}, \,z=(z_{ij}),\,\overline{z}=(\overline{z}_{ij}), z_{ij}, \overline{z}_{ij}\in\mathbb{R}$,
$$
|a_i(t,\omega, y,z)-a_i(t,\omega, \overline{y},\overline{z})|^2\le C_1 |y-\overline{y}|^2+C_2 \|z-\overline{z}\|_{\Lambda}^2, \quad i \in \mathcal{S}.
$$
\end{itemize}

\begin{proposition}\label{Nonlinear-NonMarkovian-StochThiele}
Suppose that \eqref{Nonlinear-Reserving-Payment-Functions} holds. Under the assumptions (A3) to (A5), there exists a unique solution $(Y,Z)$ to \eqref{bsde-1-Y} such that $Y$ is adapted, $Z$ is predictable and 
$$
E\left[\underset{t\in[0,T]}{\sup}|Y(t)|^2+\int_0^T\|Z(t)\|^2_{\Lambda}dt\right]< \infty.
$$
 Furthermore, the solution $(Y,Z)$ satisfies \eqref{eqn:Norberg-Thiele1}  and \eqref{expl-form-Z}.
\end{proposition}
\begin{proof}
Existence and uniqueness of a solution $(Y,Z)$ to the BSDE \eqref{bsde-1-Y} follow from Theorem 6.1. in Cohen \& Elliott (2012), using that $\gamma_{X(t-)}(t)\,dt + a_{X(t-)}(t)\, d \nu(t)=F(t) d \mu(t)$ 
for $d \mu (t) :=  dt + d\nu (t)$ and $F(t):= \mathbb{1}_{\{\nu(t)> \nu(t-)\}} a_{X(t-)}(t) + \mathbb{1}_{\{\nu(t)= \nu(t-)\}} \gamma_{X(t-)}(t)$.
Equation \eqref{bsde-1-Y} and integration-by-parts imply that $\tilde{Y}(t):=e^{-\int_0^t \delta(u) du} Y(t)$ satisfies
$$ -d\tilde{Y}(t) =e^{-\int_0^t \delta(u) du}\bigg(\gamma_{X(t-)}(t)\,dt + a_{X(t-)}(t)\, d \nu(t) -Z(t)dM(t)\bigg), \quad \tilde{Y}(T)=0,$$
which, in turn, implies that
 $$ \tilde{Y}(t)= E\Big[\int_{(t,T]} e^{-\int_0^s \delta(u)du}\Big(\gamma_{X(s-)}(s)\,ds + a_{X(s-)}(s)\, d \nu(s)\Big)\Big|\mathcal{F}^0_t\Big].$$
By multiplying the latter line with $e^{\int_0^t \delta(u) du}$ we obtain \eqref{eqn:Norberg-Thiele1}.
Finally, apply Proposition \ref{bsde-1} in order to obtain the representation for $Z$.
\end{proof}

%

Djehiche \& L\"{o}fdahl (2016) give a number of examples of life insurance contracts where the Thiele BSDE is nonlinear, including the prominent example of surrender payments. The surrender value of a life insurance at time $t$ typically equals the prospective reserve $Y(t-)$ minus a lapse fee, such that the assumptions of Proposition \ref{Nonlinear-NonMarkovian-StochThiele} hold. The prospective reserve $Y(t-)$ is also relevant if a contract is modified at time $t$. For example, if a free policy option is exercised at time $t$, then $Y(t-)$ is seen as the policyholders wealth at time $t$ which serves as a lump sum premium for the modified contract. However, in the next section we will see that Proposition \ref{Nonlinear-NonMarkovian-StochThiele} does not cover general contract modifications, such that further techniques are needed that go beyond Proposition \ref{Nonlinear-NonMarkovian-StochThiele}.

\section{Contract modifications}

Life insurance cash flows become reserve dependent upon contract modifications. At the time where a contract is changed, the current prospective reserve of the old contract is usually seen as the policyholder's wealth and is used as a lump sum premium for the new contract. In this section we model the evolution of an insurance policy as a pair of jump processes $(X,J)$,  where $X$ is the state of the policyholder and $J$ describes the different modes of the policy as a result of contract modifications.

If $(X,J)$ is a Markov process and $X$ and $J$ have no simultaneous jumps, then 
 actuarial equivalence at time zero is maintained under contract modifications if   the sum at risk upon a contract modification is zero, see e.g.~Henriksen et al.~(2014).  More precisely, the  Cantelli Theorem (cf.~Milbrodt \& Stracke (1997)) states that jumps of $J$ can be ignored in the calculation of the state-wise prospective reserves  if the sum at risk upon a jump of $J$ is zero. In this section we will generalize that concept to non-Markovian models.  The sum at risk condition will lead us to a nonlinear BSDE that, unfortunately, is  not covered by Proposition \ref{Nonlinear-NonMarkovian-StochThiele}, but we  will show  a way out based on a recursion scheme.

\subsection{State space expansion} 
Let $(X,J)$ be c{\`a}dl{\`a}g jump processes, defined on the filtered probability space $(\Omega,\F,\Ff=(\F_t)_{ t\ge 0},P)$, where $\F_t$ is the completed natural filtration of $(X,J)$ which satisfies the usual conditions.  Let $\mathcal{J}\subset \mathbb{N}_0$ be the state space of $J$, i.e.~the set of possible modes of the insurance contract. Let the
 $\Ff$-stopping  times $0=\tau_0< \tau_1< \tau_2 < \ldots $  describe the jump times of process $J$. We suppose  that $(X(0),J(0))$ is deterministic.  Moreover, we assume that  $X$ and $J$ have no simultaneous jumps.  This assumption is common in the actuarial literature for modelling lapse and contract modifications. It could  be  relaxed, but at the cost of a tedious notation, so we prefer to claim it here.

For $X$ and $J$ we define corresponding indicator processes $I^0_i(t)=\mathbf{1}_{\{X(t)=i\}}, I^1_i(t)=\mathbf{1}_{\{J(t)=i\}}$ and corresponding
counting processes 
\begin{align*}
N^0_{ij}(t)&:=\#\{s\in(0,t]:X(s-)=i, X(s)=j\},\quad N^0_{ij}(0)=0,    \\
N^1_{kl}(t)&:=\#\{s\in(0,t]:J(s-)=k, J(s)=l\},\quad N^1_{kl}(0)=0,
\end{align*}
and set
\begin{align*}
N^0(t):=\sum_{i,j : i\ne j} N^0_{ij}(t),\qquad N^1(t):=\sum_{k,l : k\ne l} N^1_{kl}(t), \quad t \geq 0.
\end{align*}
Let $\Lambda^0=(\l^0_{ij})_{ij}$ and $\Lambda^1=(\l^1_{kl})_{kl}$ denote $\mathbb{F}$-predictable   jump intensities of the processes $X$ and  $J$. 
Occasionally we will write 
$$
\l^0_{ij}(t) = \l^0_{ij}(t,J(t-)),\quad \l^1_{kl}(t) = \l^1_{kl}(t,X(t-))
$$
when the dependence of the transition intensities on the states of $J(t-)$ or $X(t-)$ shall be made explicit. This means that under each mode $k \in \mathcal{J}$, $X$ is a pure jump process with random intensities $\l^0_{ij}(t,k)$ and given each state $i\in\mathcal{S}$, $J$ is a pure jump process with random intensities $\l^1_{kl}(t,i)$.

We assume that
\begin{itemize}
\item[(A5)] 
$$
E\left[\int_0^T \Big( \sum_{i,j:i \neq j}  I^0_i(t-) \l^0_{ij}(t)+ \sum_{k,l:\, k\neq l}I^1_k(t-)\l^1_{kl}(t)  \Big)dt\right]<\infty.
$$
\end{itemize}

Since we assumed that $X$ and $J$ have no simultaneous jumps, we  can see $\widetilde{X}:=(X,J)$   as a state space expansion of the process $X$  with corresponding counting processes $((N^0_{ij})_{ij}, (N^1_{kl})_{kl})$ and associated martingales
 \begin{align}\label{mart-1}\begin{split}
M^0_{ij}(t)&=N^0_{ij}(t)-\int_0^t I^0_i(s-)\l^0_{ij}(s)\, ds,\quad M^0_{ij}(0)=0,\\
M^1_{kl}(t)&=N^1_{kl}(t)-\int_0^t I^1_k(s-)\l^1_{kl}(s)\, ds,\quad M^1_{kl}(0)=0.
\end{split}\end{align}
 That means that all results from the previous sections for the process $X$ can be transferred to the expanded jump process $\widetilde{X}:=(X,J)$. 
\begin{example}[Markovian survival model]\label{ExampleBM15}  The Markov survival models with surrender and free policy options studied in Buchardt et al.~(2015) and Buchardt \& M{\o}ller (2015) can be seen as a special class of the modulated policyholder model suggested above. As an  example, let $\mathcal{J}=\{0,1\}$ where $0$ stands for a standard policy mode and  $1$ denotes the free policy mode. Assume further that the state $X$ of the policyholder  takes values in $\cS=\{0,1,2\}$ where 0=alive, 1=dead and 2=surrender.  If we assume $\widetilde{X}:=(X,J)$ to be a Markov process with state space $\widetilde{\cS}:=\cS\times\mathcal{J}$ where
$$\begin{array}{lll}
(0,0)= \mbox{alive},\, (1,0)=\mbox{dead},\, (2,0)=\mbox{surrender},\, (0,1)=\mbox{alive free policy},\\ (1,1)=\mbox{dead free policy}, \,\, (2,1)=\mbox{surrender free policy},
\end{array}
$$
and intensities
\begin{align*}
\l^0_{01}(t,0,\omega)=\mu_{ad}(t),&& \l^0_{02}(t,0,\omega)=\mu_{as}(t),&& \l^1_{01}(t,0,\omega)=\mu_{af}(t),\\
 \l^0_{01}(t,1,\omega)=\mu^f_{ad}(t), && \l^0_{02}(t,1,\omega)=\mu^f_{as}(t),    
\end{align*}
we obtain the survival model suggested in Buchardt \& M{\o}ller (2015), Section 3.2.
\end{example}

\subsection{Modifications without actuarial equivalence} 
If maintaining of actuarial equivalence is not an objective at contract modifications, then we can simply transfer the results from Section \ref{SectionProspectiveReserves} to the expanded process $\widetilde{X}:=(X,J)$.

Suppose  that the payment process $A(t)$ is of the form
\begin{align*}\begin{array}{lll}
dA(t)=\big(\alpha_{\widetilde{X}(t-)}(t)\,dt+a_{\widetilde{X}(t-)}(t) \, d\nu(t)\big)+\sum_{i,j: i\neq j} \b_{ij }(t)\,dN^0_{ij }(t)+\sum_{k,l: k\neq l} \bar\b_{kl }(t)\,dN^1_{kl}(t),
\end{array}
\end{align*}
where $\alpha_{(i,k)}$, $a_{(i,k)}$, $\b_{ij}$ and $\bar\b_{kl}$ are $\mathbb{F}$-predictable processes which satisfy
\begin{align}\label{alpha-3}
E\bigg[ \int_{(0,T]}|a_{\widetilde{X}(t-)}(t)|^2d\nu(t)  + \int_0^T\left(|\alpha_{\widetilde{X}(t-)}(t)|^{2}+ \|\b(t)\|^2_{\Lambda^0}+\|\bar\b(t)\|^2_{\Lambda^1} \right)dt \bigg]<\infty.
\end{align}
In addition to the life insurance cash flow as defined in Section \ref{Section_Cashflow}, here we include transition payments $\bar\b_{kl}$ upon contract modifications, e.g.~a surrender payment. 
Occasionally we will write 
$$
\b^0_{ij}(t) = \b^0_{ij}(t,J(t-)),\quad \bar\b^1_{kl}(t) = \bar\b^1_{kl}(t,X(t-))
$$
when the dependence of the transition paymnets on the states of $J(t-)$ or $X(t-)$ shall be made explicit.
Setting
$$
\gamma_{(i,k)}(t):=\alpha_i(t)+\sum_{j:j\neq i} \b_{ij}(t)\l^0_{ij}(t)+
\sum_{l:l\neq k} \bar\b_{kl}(t)\l^1_{kl}(t,i)
$$
and using the martingales associated with $((N^0_{ij})_{ij},((N^1_{kl})_{kl})$, the prospective reserve at time $t$ satisfies
\begin{align}\label{reserve-J}
Y(t)=E\Big[\int_{(t,T]} e^{-\int_t^s \delta(u)du}\left(\gamma_{\widetilde{X}(s-)}(s)\,ds+a_{\widetilde{X}(s-)}(s)d\nu(s) \right)\Big|\F_t\Big].
\end{align}
By applying the results from Section \ref{SectionProspectiveReserves}  on the expanded state space process $\widetilde{X}=(X,J)$, we can show that the prospective reserve \eqref{reserve-J} is the unique solution of the BSDE
\begin{align}\label{bsde-1-Y-J}\begin{array}{lll}
&dY(t)=\left(-\delta(t)Y(t)+\gamma_{\widetilde{X}(t-)}(t)\right)dt+a_{\widetilde{X}(t-)}(t)d\nu(t) +Z^0(t)dM^0(t)+Z^1(t)dM^1(t),\\
& Y(T)=0,
\end{array}
\end{align}
where  $Z^0=(Z^0_{ij}, \, i\neq j)$ and $Z^1=(Z^1_{kl}, \, k\neq l)$ are unique predictable processes  satisfying
\begin{equation}\label{Z}
E\left[\int_0^T \left(\|Z^0(s)\|^2_{\Lambda^0}+\|Z^1(s)\|^2_{\Lambda^1}\right)ds\right]<\infty.
\end{equation}
Since $\mathbb{F}$ is the natural filtration of $(X,J)$ and  the two processes have no simultaneous jumps,  by following the arguments in the proof of Proposition \ref{bsde-1}  we can show that  $Z^0$ and $Z^1$ almost surely satisfy
\begin{align}\begin{split}\label{Z0Z1Def}
  I^0_i(t-) Z^0_{ij}(t)&=  \sum_k I^0_i(t-) I^1_k(t-) \Big(E[ Y(t) |\mathcal{F}_{t-}, \widetilde{X}(t)=(j,k)]- E[ Y(t) |\mathcal{F}_{t-}, \widetilde{X}(t)=(i,k)]\Big),\\
  I^1_k(t-)Z^1_{kl}(t)&=  \sum_i I^0_i(t-) I^1_k(t-) \Big(E[ Y(t) |\mathcal{F}_{t-}, \widetilde{X}(t)=(i,l)]- E[ Y(t) |\mathcal{F}_{t-},\widetilde{X}(t)=(i,k)]\Big)
\end{split}\end{align}
for each $t>0$.

\begin{example}[The Markovian case and Thiele's differential equation]
Suppose that 
\begin{align*}
 a_i(t)=0,\quad 
 \b_{ij}(t)=\b_{ij}(t,J(t-)),\quad  \bar\b_{kl}(t)=\bar\b_{kl}(t,X(t-)),
\end{align*}
and let for each $i,j \in \mathcal{S}$, $i\neq j$, and $k,l\in  \mathcal{J}$, $k \neq l$, the payment processes
$\alpha_{(i,k)}(t)$,  $\b_{ij}(t,k)$, $\bar\b_{kl}(t,i)$ and the transition intensities  $\l^0_{ij}(t,k)$  and $\l^1_{kl}(t,i)$ be deterministic functions in $t$.  Assume further that the discount factor $\delta$ is deterministic and continuous in $t$. Then the process $\widetilde{X}=(X,J)$ is a Markov process and  the prospective reserve \eqref{reserve-J} becomes
\begin{equation*}
Y(t)=E\Big[\int_{t}^{T} e^{-\int_t^s \delta(u)du}\gamma_{\widetilde{X}(s-)}(s)\,ds\Big| X(t),J(t)\Big]=V(t,X(t), J(t))
\end{equation*}
for some deterministic function $V:\,[0,T] \times\cS \times \mathcal{J}\rightarrow \R$.
In particular,  we may apply the Feynman-Kac's formula (cf.\@ Lemma 2.1 in Djehiche \& L\"{o}fdahl (2016)) to see that the function
\begin{equation*}
V(t,i,k)=E\Big[\int_t^T e^{-\int_t^s \delta(u)du}\,\gamma_{\widetilde{X}(s-)}(s)\,ds\Big|X(t)=i, J(t)=k\Big]
\end{equation*}
is differentiable in $t$ and satisfies the following ordinary differential equation
\begin{equation}\label{Norberg-Thiele-J-3}
\left\{\begin{array}{lll}
\frac{d}{dt}V(t,i,k)=\delta(t)V(t,i,k)-\gamma_{(i,k)}(t)-Q_{0}V(t,i,k)-Q_{1}V(t,i,k),
\\ V(T,i,k)=0, \quad (i,k)\in\cS \times \mathcal{J},
\end{array}
\right.
\end{equation}
where
$$\begin{array}{lll}
Q_{0}V(t,i,k)=\sum_{j: j\neq i} \l^0_{ij}(t,k)(V(t,j,k)-V(t,i,k)), \\ 
Q_{1}V(t,i,k)=\sum_{l:l \neq k} \l^1_{kl}(t,i)(V(t,i,l)-V(t,i,k)),
\end{array}
$$
which includes a modulated version of the celebrated Thiele equation. Indeed, in terms of the  modulated sum-at-risk, in mode $k$, assuming $\bar\b_{kl}=0$, 
$$
R_{ij}(t,k):=\beta_{ij}(t,k)+V(t,j,k)-V(t,i,k),
$$
the equation \eqref{Norberg-Thiele-J-3} can be rearranged to take the form
$$\left\{\begin{array}{lll}
\frac{d}{dt}V(t,i,k)=\delta(t)V(t,i,k)-\alpha_{(i,k)}(t)-Q_{1}V(t,i,k)-\sum_{j:j\neq i} R_{ij}(t,k)\l^0_{ij}(t,k)=0, \\
 V(T,i,k)=0,\quad (i,k)\in\mathcal{S}\times \mathcal{J}.
\end{array}
\right.
$$
\end{example}

\medskip
\subsection{Modifications that maintain actuarial equivalence} 
Actuarial equivalence is maintained upon a contract modification at random time $\tau$ if the prospective reserve on $[0,\tau)$ is unaffected by the modification. According to the actuarial literature (see e.g.~Henriksen et al.~(2014)), in Markov models this can be achieved  by making sure that the sum-at-risk for the contract modification equals zero. In this section we generalize that concept to non-Markovian models.

\begin{proposition}\label{Construction_measrures_Pm}
For each $m \in \mathbb{N}_0$   there exists a unique probability measure $P^m$ on $(\Omega, \mathcal{F})$ such that the bivariate jump process $(X, J)$ has transition rates of $(\lambda^0_{ij},\kappa_m \lambda^1_{kl})_{ij,kl}$, where $\kappa_{m}(t):= \mathbb{1}_{\{ t \leq \tau_{m}  \}}$, $t \geq 0$. Moreover,  it holds that
\begin{enumerate}[(a)]
    \item  $P^m= P$ on $\mathcal{F}_{\tau_{m}}$,
    \item $P^m \sim P$ on $\mathcal{F}_{\tau_{m+1}-}$,
    \item $P^m \ll  P$ on  $\mathcal{F}_{\infty}$.
\end{enumerate}
\end{proposition}
\begin{proof}
Because of Assumption (A5) there exists a $P$-zero set $N \in \mathcal{F}$ such that 
\begin{align*}
    \int_0^T\sum_{i,j: i\neq j}(I^0_i(t-) \l^0_{ij}(t)+ I^1_{k}(t-)\l^1_{kl}(t))dt  < \infty
\end{align*}
on $\Omega \setminus N$. Without loss of generality we may redefine $\l^0_{ij}$ and $\l^1_{kl}$ such that the latter inequality holds on all of $\Omega$. 
By applying Proposition \ref{ModelFromCompensators} on the process $\widetilde{X}$, we obtain that there is a unique  probability measure  $P^m$ such that  $\widetilde{X}$ has the transition intensities $(\lambda^0_{ij},\kappa_m \lambda^1_{kl})_{ij,kl}$ w.r.t.~the $P^m$-completed natural filtration of $X$.

According to Theorem 10.2.6 in Last \& Brandt (1995), we necessarily have $P^m \ll P$ on $\sigma(X(s):s\geq 0)$, which implies that the $P^m$-completion comprises the $P$-completion of the natural filtration of $X$. Moreover,  Corollary 10.2.7 in Last \& Brandt (1995) gives an explicit formula for the Radon-Nikodym derivative of $P^m$  with respect to $ P$. This Radon-Nikodym derivative equals $1$ on $\mathcal{F}_{\tau_{m}}$ and is
strictly positive  on  $\mathcal{F}_{\tau_{m+1}-}$.
\end{proof}

A soon as the $m$-th contract modification occurs at time $\tau_m$, the transition intensities $(\kappa_m \lambda^1_{kl})_{kl}$ for jumps of $J$ equal zero such that no further contract modifications can happen. Thus,  the filtered probability space $(\Omega, \mathcal{F}, P^m, \mathbb{F})$ describes a life insurance model where at most $m$ contract modifications occur. 

\begin{theorem}[Cantelli Theorem for non-Markovian models]\label{CantelliTheorem}
For $m \in \mathbb{N}_0$ let $(Y,Z^0,Z^1)$ and $(Y^m,Z^{0,m},Z^{1,m})$ be the unique solutions of BSDE \eqref{bsde-1-Y-J} under the probability measures $P$ and $P^m$, respectively,  for  $P^m$ defined as in Proposition \ref{Construction_measrures_Pm}. 
Then we have
\begin{align*}
(Y(t),Z^0(t),Z^1(t))=(Y^m(t),Z^{0,m}(t),Z^{1,m}(t)) \quad &P\textrm{-a.s.~for all }t  \in [0, \tau_{m+1})
\end{align*}
if and only if
\begin{align}\label{SumAtRiskCondition}
    \mathbb{1}_{\{\tau_m <  t \leq \tau_{m+1}\}} \sum_{l:l\neq k} \Big( \bar{\beta}_{kl}(t) + Z^1_{kl}(t)\Big) I^1_k(t-)  \lambda^1_{kl}(t)   =0 
\end{align}
$d P\times dt$ -a.e.~for all  $k \in \mathcal{J}$.
\end{theorem}
\begin{proof}
By construction, the $P^m$-compensator of $N^1$ is zero on $(\tau_{m},\infty)$, which implies that $P^m(\tau_{m+1} < \infty) =0$.

If  \eqref{SumAtRiskCondition} holds, then the definition of $P^m$ implies that  $(Y,Z^0,Z^1)$ and $(Y^m,Z^{0,m},Z^{1,m})$ solve the same BSDE on $[0,\tau_{m+1})$, thus they are $P^m$-almost surely equal on $[0,T]$, using that $P^m(T < \tau_{m+1})=1$ and  the $P^m$-  uniqueness of the BSDE solution  $(Y^m,Z^{0,m},Z^{1,m})$. Moreover, since  $P^m|_{\mathcal{F}_{\tau_{m+1}-}} \sim P|_{\mathcal{F}_{\tau_{m+1}-}}$,  the processes $(Y,Z^0,Z^1)$ and $(Y^m,Z^{0,m},Z^{1,m})$ are $P$-a.s.~equal on $[0,\tau_{m+1})$.

On the other hand, if we know that $(Y(t),Z^0(t),Z^1(t))=(Y^m(t),Z^{0,m}(t),Z^{1,m}(t))$ for all $t < \tau_{m+1}$, then the difference on  $[0, \tau_{m+1})$ of the corresponding BSDEs is of the form
\begin{align*}
    0 =   (1-\kappa_m(t)) \sum_{k,l:l\neq k} \Big( \bar{\beta}_{kl}(t) + Z^1_{kl}(t) \Big) I^1_k(t-) \lambda^1_{kl}(t) dt .
\end{align*}
Since the events $\{J(t-)=k\}$, $k \in \mathcal{J}$, are mutually exclusive, the latter equation is equivalent to condition \eqref{SumAtRiskCondition}.
\end{proof}

Since $\tau_{m+1} > 0$ for all $m \in \mathbb{N}_0$, Theorem \ref{CantelliTheorem} describes a situation where  the net premium condition at time zero is not affected by contract modifications.

By applying Proposition \ref{bsde-1} on the extended random pattern of states $\widetilde{X}=(X,J)$, we can see that 
the factor $\bar{\beta}_{kl}(t) + Z^1_{kl}(t)$ in  \eqref{SumAtRiskCondition} represents the sum-at-risk for a transition of $J$ from $k$ to $l$ at time $t$. Hence, Theorem \ref{CantelliTheorem} is a non-Markovian generalization of the Cantelli Theorem. 

Now we are seeking to construct life insurance policies that satisfy condition \eqref{SumAtRiskCondition}. A common approach in insurance practice is to start with a given cash flow $A$ and to add adjustment factors (or scaling factors) in such a way that actuarial equivalence is maintained upon contract modifications. Based on the life insurance cash flow $A$, we define an adjusted cash flow $\hat{A}$ by 
\begin{align*}
   d\hat{A}(t)  =  \sum_{m=0}^{\infty} \mathbb{1}_{\{\tau_m <  t \leq \tau_{m+1}\}}  \rho_m\, dA(t) ,\quad \hat{A}(0)=A(0),
\end{align*}
where $\rho_0:=0$ and $\rho_m$, $m \in \mathbb{N}$, are $\mathcal{F}_{\tau_m}$-measurable random variables.  We interpret $\rho_m$ as an actuarial adjustment (or scaling) that is applied on the future life insurance cash flow upon the $m$-th contract modification.
In particular, the transition payments of $\hat{A}$ upon a jump of $J$ are of the form
\begin{align*}
 \hat{\bar{\beta}}_{kl}(t):= \sum_{m=0}^{\infty} \mathbb{1}_{\{\tau_m <  t \leq \tau_{m+1}\}}  \rho_m\,  \bar{\beta}_{kl}(t).
\end{align*}
We can represent the adjustment factors by  $\rho_m=\rho_m(\tau_m,\widetilde{X}(\tau_m))$ for mappings $\rho_m(t,(i,k))(\omega)$, $m \in \mathbb{N}$,  that are jointly measurable in $(t,(i,k),\omega) \in [0,\infty) \times (\mathcal{S} \times \mathcal{J}) \times \Omega$ and such that $\omega \mapsto \rho_m(t,(i,k))(\omega) $ is  $\mathcal{F}_{\tau_m-}$-measurable  for each $(t,(i,k))$.

In the following proposition we pretend that we have a life insurance model with up to $m$  contract modifications and known adjustment factors $\rho_1, \ldots, \rho_m$ and aim to expand the life insurance model to a maximum of $m+1$ contract modifications. We give  a condition for $\rho_{m+1}$  that implies \eqref{SumAtRiskCondition}, i.e.~the condition ensures actuarial equivalence upon the $(m+1)$-th contract modification. 
\begin{proposition} \label{LemmaSufficCondRho}
For $m \in \mathbb{N}_0$ suppose that  $(Y^{m+1}, Z^{0,m+1}, Z^{1,m+1})$ and $(\hat{Y}^{m+1}, \hat{Z}^{0,m+1}, \hat{Z}^{1,m+1})$ are  unique solutions to the BSDE \eqref{bsde-1-Y-J} w.r.t.~the life insurance cash flows $A$ and  $\hat{A}$, respectively, and w.r.t.~the probability measure $P^{m+1}$ as defined in Proposition \ref{Construction_measrures_Pm}. 
Then for each $k,l \in \mathcal{J}$, $k\neq l$, we have
\begin{align}\label{SumAtRiskIsZero}
\mathbb{1}_{\{\tau_{m} <  t \leq \tau_{m+1}\}} I^1_k(t-)  \big(\hat{\bar{\beta}}_{kl}(t) + \hat{Z}^{1,m+1}_{kl}(t)\big) =0 
\end{align}
if and only if
\begin{align}\label{rho_t_i_l_condition}
    \rho_{m+1}(t,(i,l)) 
    = \frac{\hat{Y}^{m+1}(t-) - \rho_m\, 
    \bar{\beta}_{J(t-)l}(t)}{E^{m+1}[  Y^{m+1}(t)  | \mathcal{F}_{\tau_{m+1}-}, \tau_{m+1}=t, \widetilde{X}(\tau_{m+1})=(i,l)]}, \quad \tau_m <  t \leq \tau_{m+1}, 
\end{align}
under the convention $0/0:=1$.
\end{proposition}
\begin{proof}
By arguing analogously to the proof of Proposition \ref{bsde-1}, we can show that the left hand side of \eqref{SumAtRiskIsZero} almost surely equals
\begin{align*}
      &\mathbb{1}_{\{\tau_m <  t \leq \tau_{m+1}\}} I^1_k(t-)\sum_i I^0_i(t-) \bigg(  \hat{\bar{\beta}}_{kl}(t) + E^{m+1}[\hat{\zeta}(t-)| \mathcal{F}_{\tau_{m+1}-}, \tau_{m+1}=t, \widetilde{X}(\tau_{m+1})=(i,l) ] \\
      &\qquad \qquad \qquad \qquad \qquad - \frac{E^{m+1}[ \hat{\zeta}(t-) \mathbb{1}_{\{\widetilde{X}(t)=\widetilde{X}(t-)\}} | \mathcal{F}_{t-}] }{E^{m+1}[  \mathbb{1}_{\{\widetilde{X}(t)=\widetilde{X}(t-)\}} | \mathcal{F}_{t-}]}\bigg)
\end{align*}
for
\begin{align*}
   \hat{\zeta}(t):= \int_{(t,T]} e^{-\int_t^s \delta(u)du}\sum_m \mathbb{1}_{\{\tau_m <  t \leq \tau_{m+1}\}}   \rho_m(\tau_m,\widetilde{X}(\tau_m)) \Big( \gamma_{\widetilde{X}(s-)}(s)\,ds  +a_{\widetilde{X}(s-)}(s)\, d\nu(s)\Big).
 \end{align*}
Since the compensators of $N^0$ and $N^1$ have Lebesgue-densities, we can show that 
\begin{align*}
    \frac{E^{m+1}[ \hat{\zeta}(t-) \mathbb{1}_{\{\widetilde{X}(t)=\widetilde{X}(t-)\}} | \mathcal{F}_{t-}] }{E^{m+1}[  \mathbb{1}_{\{\widetilde{X}(t)=\widetilde{X}(t-)\}} | \mathcal{F}_{t-}]} = E^{m+1}[ \hat{\zeta}(t-)  | \mathcal{F}_{t-}]=\hat{Y}^{m+1}(t-).
\end{align*}
On the other hand,  since 
$$\hat{\zeta} (\tau_{m+1}-)=  \rho_{m+1}(\tau_{m+1},\widetilde{X}(\tau_{m+1})) \,\zeta (\tau_{m+1}) + \rho_{m}\, a_{\widetilde{X}(\tau_{m+1}-)}(\tau_{m+1})\,\Delta v(\tau_{m+1})$$ 
$P^{m+1}$-almost surely, and since the existence of Lebesgue-densities for the compensators of $N^1$ implies $P^{m+1}(\Delta \nu(\tau_{m+1})=1)=0$, we obtain
\begin{align*}
    &E^{m+1}[\hat{\zeta}(t-)| \mathcal{F}_{\tau_{m+1}-}, \tau_{m+1}=t, \widetilde{X}(\tau_{m+1})=(i,l) ] \\
    & = \rho_{m+1}(t,(i,l))  \, E^{m+1}[ Y^{m+1}(t)  | \mathcal{F}_{\tau_{m+1}-}, \tau_{m+1}=t, \widetilde{X}(\tau_{m+1})=(i,l) ].
\end{align*}
Altogether, we can conclude that \eqref{SumAtRiskIsZero} is equivalent to
\begin{align*}
0= &\mathbb{1}_{\{\tau_{m} <  t \leq \tau_{m+1}\}} I^1_k(t-) \Big( \rho_m\,\bar{\beta}_{kl}(t) 
-\hat{Y}^{m+1}(t-) \\
& \qquad \qquad+ \rho_{m+1}(t,(i,l))   E^{m+1}[ Y^{m+1}(t)  | \mathcal{F}_{\tau_{m+1}-}, \tau_{m+1}=t, \widetilde{X}(\tau_{m+1})=(i,l) ] \Big).
\end{align*}
\end{proof}

In formula \eqref{rho_t_i_l_condition}, $\hat{Y}^{m+1}(t-)$ gives the policyholders wealth just before a contract modification, from which we deduct the modification lump sum payment $ \rho_m \bar{\beta}_{J(t-)l}(t)$. The denominator is the value of the new contract before  actuarial adjustments. 

Suppose for the moment that we have a life insurance model where at most one contract modification can occur at time $\tau:=\tau_1$, i.e.~we have $P=P^1$.  In the situation of Proposition \ref{LemmaSufficCondRho} we obtain then that  the adjustment factor $\rho:= \rho_1$ can be represented as 
\begin{align*}
\rho = f(\tau,\hat{Y}(\tau-)) 
\end{align*}
for a jointly measurable mapping $f$ such that $f(t,y)$ is $\mathcal{F}_{t-}$-adapted for each $t> 0$. As a result, the process  $(\hat{Y},\hat{Z}^0,\hat{Z}^1)$ corresponds to a nonlinear BSDE of the form \eqref{bsde-1-Y-J} but with payment process 
\begin{align*}
    \hat{\gamma}_{\widetilde{X}(t-)}(t) = \mathbb{1}_{\{t <\tau\}} \gamma_{\widetilde{X}(t-)}(t)+  \mathbb{1}_{\{t \geq \tau\}}f(\tau,\hat{Y}(\tau-)) \gamma_{\widetilde{X}(t-)}(t), \quad t \geq 0.
\end{align*}
Unfortunately, Proposition \ref{Nonlinear-NonMarkovian-StochThiele} does not apply here since for $t > \tau$ the process $\gamma_{\widetilde{X}(t-)}(t)$ depends on the further past of  $\hat{Y}$ via $\hat{Y}(\tau)$  rather than $\hat{Y}(t-)$.  In the literature we can find existence and uniqueness results also for BSDEs of such kind, see e.g.~Cheridito \& Nam (2017), but they come with very restrictive Lipschitz assumptions that are usually not satisfied in our setting. Therefore,  we present now an alternative way for calculating the adjustment factors $\rho_m$, $m\in \mathbb{N}$.

\begin{theorem}[Recursive calculation of adjustment factors]\label{RecursiveCalcTheorem}
For each $m \in \mathbb{N}_0$ let $(Y^m, Z^{0,m}, Z^{1,m})$ and $(\hat{Y}^m, \hat{Z}^{0,m}, \hat{Z}^{1,m})$ be the unique solutions of the BSDE \eqref{bsde-1-Y-J} under the probability measure $P^m$ w.r.t.~the life insurance cash flows $A$ and $\hat{A}$, respectively.
If
\begin{align}\label{RecEqForRho}
    \rho_{m+1}= \frac{\hat{Y}^{m}(\tau_{m+1}-)- \rho_{m}\, \beta^1_{J(\tau_{m})J(\tau_{m+1})}(\tau_{m+1})}{Y^{m+1}(\tau_{m+1})}, \quad m \in \mathbb{N}_0,
\end{align}
 then 
\begin{align*}
    \hat{Y}(t) = \hat{Y}^m(t), \quad  t < \tau_{m+1},\, m \in \mathbb{N}_0, 
\end{align*}
where $(\hat{Y}, \hat{Z}^{0}, \hat{Z}^{1})$ is the unique solution of the BSDE \eqref{bsde-1-Y-J} under the probability measure $P$ w.r.t.~the life insurance cash flows $\hat{A}$.
\end{theorem}
\begin{proof}
The absolute continuity of the compensators of the counting processes $(N^0,N^1)$ implies that $P(\Delta \nu(\tau_{m+1})=1)=0$ and
\begin{align*}
    \hat{Y}^{m+1}(t-)&=  \hat{Y}^{m+1}(t) +  \sum_{n} \mathbb{1}_{\{\tau_n < t \leq \tau_{n+1}\}}  \rho_n \Big(   a_{\widetilde{X}(t-)}(t)\, \Delta v(t) +\sum_{k,l:k\neq l} \beta^1_{kl}(t)\, \Delta N^1_{kl} (t)\Big).
\end{align*}
These two facts and $\hat{Y}^{m+1}(\tau_{m+1})=\rho_{m+1} Y^{m+1}(\tau_{m+1})$ yield
\begin{align*}
    \hat{Y}^{m+1}(\tau_{m+1}-)&=  \rho_{m+1} Y^{m+1}(\tau_{m+1}) +\rho_{m} \beta^1_{J(\tau_{m})J(\tau_{m+1})}(\tau_{m+1}).
\end{align*}
Moreover, replacing $\rho_{m+1}$ by \eqref{RecEqForRho} leads to 
\begin{align*}
    \hat{Y}^{m+1}(\tau_{m+1}-)&=  \hat{Y}^{m}(\tau_{m+1}-).
\end{align*}
Thus, condition \eqref{RecEqForRho} is equivalent to 
\begin{align}\label{RecEqForRho2}
    \rho_{m+1}= \frac{\hat{Y}^{m+1}(\tau_{m+1}-)- \rho_{m}\, \beta^1_{J(\tau_{m})J(\tau_{m+1})}(\tau_{m+1})}{Y^{m+1}(\tau_{m+1})}, \quad m \in \mathbb{N}_0.
\end{align}
Hence, $\hat{A}$ is a cash flow whose adjustment factors can be represented as $\rho_{m+1}=\rho_{m+1}(\tau_{m+1},\widetilde{X}(\tau_{m+1}))$, $ m\in \mathbb{N}_0$, for mapping $\rho_{m+1}(t,(i,l)$ defined as by \eqref{rho_t_i_l_condition}. According to Proposition \ref{LemmaSufficCondRho} equation \eqref{SumAtRiskIsZero} holds, and by applying Theorem \ref{CantelliTheorem} we obtain  $\hat{Y}^{m+1}(t)= \hat{Y}^{m}(t)$ for all $t \in [0, \tau_{m+1})$. In particular, this implies that  $\hat{Y}^k(t)= \hat{Y}^{m+1}(t)$  on  $[0, \tau_{m+1})$ for all $k \geq m+1$.
Because of Assumption (A5) the paths of $N^1$ have at most finitely many jumps on $[0,T]$. Thus, for almost each $\omega \in \Omega$ there exists an $n_0$ such that $\tau_{n_0}(\omega) > T$. Hence, for almost all $\omega \in \Omega$ the sequence $(\hat{Y}^n(\omega),\hat{Z}^{0,n}(\omega),\hat{Z}^{1,n}(\omega))_n $ is for $n \geq n_0$, so it converges almost surely for $n \rightarrow \infty$ to a limit  $(Y^*, Z^{0,*}, Z^{1,*})$. This  limit equals $(\hat{Y}, \hat{Z}^{0}, \hat{Z}^{1})$ since it solves the same BSDE.
\end{proof}

While in  \eqref{rho_t_i_l_condition} the right hand side depends on $\hat{Y}^{m+1}$, in  \eqref{RecEqForRho} we just need to know $\hat{Y}^{m}$.  This allows for a recursive calculation of the adjustment factors, starting at $m=0$. In the $m$-th recursion step we calculate $\hat{Y}^{m}$ by solving a BSDE that is linear, since $\rho_0, \ldots, \rho_m$ are already known from the previous recursion steps.  

Theorem \ref{RecursiveCalcTheorem} implies that $\hat{Y}$, $\hat{Y}^m$ and $\hat{Y}^0$ are identical at time zero  for all $m \in \mathbb{N}_0$, which means that  the net premium condition at time zero is unaffected by the contract modifications.

\begin{example}\label{Example43}
 Consider the survival model suggested in Buchardt \& M{\o}ller (2015), Section 3, where the state $X$ of the policyholder takes values in $\cS=\{0,1,2\}$ where 0=alive, 1=dead and 2=surrender. Let $\mathcal{J}=\{0,1\}$ where $0$ stands for a standard policy mode and  $1$ denotes the free policy mode.

$\bullet$ When the policy is in mode $0$, the payments consist of a benefit rate $b(t)$, a premium rate $\pi(t)$ and a payment $b_{ad}(t)$ upon death at time $t$, i.e.~we have
$$
\a_{(0,0)}(t)=b(t)-\pi(t),\quad \b^0_{01}(t,0)=b_{ad}(t).$$
Payment upon surrender at time $t$ is
\begin{equation}\label{surr-1}
\b^0_{02}(t,0)=(1-\kappa)Y^{0}(t),
\end{equation}
where $\kappa$ is a given constant in $[0,1]$.
Therefore,
$$
\gamma_{(0,0)}(t)=b(t)-\pi(t)+\l^0_{01}(t,0)b_{ad}(t)+\l^0_{02}(t,0)(1-\kappa)Y^{0}(t).
$$

\medskip
$\bullet$ When the policy is in mode $1$, the free policy regime, the premiums $\pi(t) $ are  waived and the benefits are  reduced by the adjustment factor $\rho_1$,
$$
\gamma_{(0,1)}(t)=\rho_1 \left(b(t)+\l^0_{01}(t,1)b_{ad}(t)+\l^0_{02}(t,1)(1-\kappa)Y^{1}(t)\right),
$$
where the third addend in the bracket is the payment upon surrender $\b^0_{12}(t,1)=(1-\kappa)Y^{1}(t)$ in mode $1$.

Following Buchardt \& M{\o}ller (2015), the adjustment factor $\rho_1$ is determined by
$$
\rho_1 =\frac{Y^{0}(\tau_1)}{Y^{1}(\tau_1)},
$$
which is equivalent to \eqref{RecEqForRho} since the payments processes $\bar\b_{kl}$ are zero here and the process $Y^0$ has  continuous paths. 
Consequently, according to Theorem \ref{RecursiveCalcTheorem},  $Y^0$ equals $Y^1$ on $[0,\tau_1)$. In particular, if $Y^0$ satisfies the net premium condition at time zero, then $Y^1$ satisfies it as well.
\end{example}

\end{document}